\pdfoutput=1
\documentclass[journal]{IEEEtran}

\usepackage[T1]{fontenc}
\usepackage[utf8]{inputenc}
\usepackage{cite}
\usepackage{amsmath,amssymb}
\usepackage{graphicx}
\usepackage{xurl}
\usepackage[hidelinks]{hyperref}
\usepackage{orcidlink}
\newtheorem{assumption}{Assumption}
\newtheorem{proposition}{Proposition}

\begin{document}

\title{PQLN: Post-Quantum Security for the Bitcoin Lightning Network's Off-Chain Surfaces}

\author{Ahmet~Kurt\,\orcidlink{0000-0002-7175-1739}, Abdul-Salem~Beibitkhan\,\orcidlink{0009-0002-1799-2213}, 
Yacoub Hanna\,\orcidlink{0000-0002-8954-9336},
and~Abdullah~Aydeger\,\orcidlink{0000-0003-3333-1941}%
\thanks{\emph{(Corresponding author: Ahmet Kurt.)}}%
\thanks{Ahmet Kurt is with East Texas A\&M University, RELLIS Campus, Bryan, TX 77807 USA (e-mail: ahmet.kurt@etamu.edu).}%
\thanks{Abdul-Salem Beibitkhan is with North American University, Stafford, TX 77477 USA (e-mail: abyeibitkhan@na.edu).}%
\thanks{Yacoub Hanna is with Florida Gulf Coast University, Fort Myers, FL 33965 USA (e-mail: yhanna@fgcu.edu).}%
\thanks{Abdullah Aydeger is with the Florida Institute of Technology, Melbourne, FL 32901 USA (e-mail: aaydeger@fit.edu).}%
}

\markboth{}%
{Kurt \MakeLowercase{\textit{et al.}}: PQLN: Post-Quantum Security for the Bitcoin Lightning Network's Off-Chain Surfaces}

\maketitle

\begin{abstract}
A cryptographically relevant quantum computer will break the elliptic-curve cryptography behind Bitcoin and its Lightning Network, the most widely used payment channel network. Even a post-quantum consensus upgrade of Bitcoin would not cover Lightning's off-chain surfaces, so its gossip, peer transport, invoices, payment onions, and offers need separate protection. An adversary can already record Lightning's encrypted traffic today and decrypt it once such a computer exists. Lightning can therefore move to post-quantum cryptography now, without waiting for Bitcoin, and stop such harvest-now-decrypt-later attacks along with node impersonation, invoice forgery, and payment deanonymization. In this paper, we propose PQLN, a hybrid post-quantum extension of Lightning that protects all of these surfaces with the lattice-based standards ML-DSA and ML-KEM. PQLN distributes post-quantum node identities through Lightning's gossip, hybridizes the transport handshake, adds post-quantum signatures to invoices, commits post-quantum keys in offers, and makes payment onions hybrid. Since post-quantum material is much larger than its elliptic-curve counterpart, we introduce techniques that fit it into Lightning's existing message formats and size limits. We analyze the security of PQLN against a quantum adversary and implement it in rust-lightning, a major Lightning implementation. Our evaluation with real Lightning nodes shows that PQLN nodes interoperate with unmodified nodes. The added cryptographic operations take at most 0.33~milliseconds, and the main cost is communication, since gossip data grows about tenfold with ML-DSA and about fourfold with the smaller Falcon. To our knowledge, PQLN is the first post-quantum design, implementation, and evaluation for Lightning.
\end{abstract}


\section{Introduction}
\label{sec:introduction}

\IEEEPARstart{P}{ublic-key} cryptography in use today rests on the hardness of integer factorization and discrete logarithms, and Shor's algorithm solves both problems in polynomial time on a sufficiently large quantum computer~\cite{shor1997}. Recent advances have narrowed the gap between theory and practice, since Google's Willow processor demonstrated a logical qubit whose error rate decreases as more physical qubits are added~\cite{willow2025}, while estimates of the resources required to break 256-bit elliptic-curve keys continue to fall~\cite{babbush2026, chevignard2026}. Moreover, adversaries can record encrypted traffic today and decrypt it once a cryptographically relevant quantum computer exists, a strategy known as \textit{harvest-now-decrypt-later}~\cite{mosca2018, nistir8547}. In response, the National Institute of Standards and Technology (NIST) published its first post-quantum (PQ) cryptography standards in August 2024, including the Module-Lattice-Based Key-Encapsulation Mechanism (ML-KEM) and the Module-Lattice-Based Digital Signature Algorithm (ML-DSA)~\cite{fips203, fips204}. NIST also plans to disallow elliptic-curve signatures and key exchange after 2035~\cite{nistir8547}.

These concerns are especially pressing for cryptocurrencies since they rely on elliptic-curve cryptography to protect high-value digital assets. Bitcoin~\cite{nakamoto2008} is the most consequential case because it accounts for more than half of the market capitalization of all cryptocurrencies.\footnote{\url{https://coinmarketcap.com/charts/bitcoin-dominance/}} Bitcoin authorizes transactions with secp256k1 signatures, and a quantum adversary that observes a public key could use Shor's algorithm to recover the corresponding private key~\cite{aggarwal2018}. Indeed, around 6.9 million BTC, more than a third of the circulating supply, reside in outputs with exposed public keys~\cite{babbush2026}. However, addressing this vulnerability at the Bitcoin base layer is difficult because replacing the underlying cryptography requires a consensus upgrade, and the current PQ proposals remain at an early stage~\cite{bip360}.

The Lightning Network (LN)~\cite{poon2016} is a second-layer protocol that turns Bitcoin into an instant payment system by moving payments into off-chain channels. At the time of writing, LN is the most widely used payment channel network with around 16,000 nodes and 33,000 public channels holding roughly 3750 BTC.\footnote{\url{https://mempool.space/lightning}} LN also entered mainstream use when Coinbase enabled Lightning transfers in 2024.\footnote{\url{https://www.coinbase.com/blog/coinbase-integrates-bitcoins-lightning-network-in-partnership-with}} However, LN's own protocol stack is as quantum-exposed as the base layer. Nodes authenticate gossip with the Elliptic Curve Digital Signature Algorithm (ECDSA), derive the keys of every peer connection from Elliptic Curve Diffie-Hellman (ECDH), sign payment requests with ECDSA and Schnorr signatures, and build payment onions from one ECDH secret per hop. All of these operations rest on the hardness of the elliptic-curve discrete logarithm problem, which Shor's algorithm solves. A quantum adversary can therefore impersonate nodes, decrypt recorded transport sessions, forge payment requests, and deanonymize payments. Despite this exposure, LN has no PQ defense to date, and the only prior examination is a forum analysis that stops short of a design~\cite{roasbeef2026pq}.

However, only part of Lightning's exposure is rooted in Bitcoin. The keys and signatures behind channel funds live inside Bitcoin transactions, so they cannot change without a fork. Everything else lives purely off-chain in messages between Lightning nodes, so these off-chain surfaces can adopt PQ cryptography through a software update alone. They also stay exposed after Bitcoin adopts PQ outputs, because a base-layer upgrade does not cover them. In other words, Lightning has to protect these surfaces itself, and it can do so before Bitcoin completes its own transition.

In this paper, we propose PQLN, a PQ extension of Lightning that protects its off-chain surfaces. PQLN follows a hybrid conservative-extension approach, meaning that every classical mechanism stays in place with ML-DSA and ML-KEM added alongside it, so Lightning's security remains intact. However, hybridizing the primitives is only the starting point, because their deployment faces four obstacles. First, Lightning has no certificate infrastructure to distribute PQ keys. PQLN therefore turns LN's gossip into a key distribution channel, where every node announces its PQ keys and others pin them on first sight. For payees absent from gossip, a reusable payment offer additionally commits a fresh signing key. Second, a quantum adversary can forge any negotiation message, so we run the hybrid transport handshake on a dedicated port without in-band negotiation and check for downgrades against pinned or offer-committed keys. Third, PQ keys, signatures, and ciphertexts are much larger than their classical counterparts and do not fit into Lightning's messages. For example, a 2420-byte ML-DSA-44 signature has to fit into invoice fields of at most 639 bytes each. Similarly, one 1088-byte ML-KEM-768 ciphertext would nearly fill the 1300-byte onion payload, and a payment needs one per hop. PQLN therefore splits oversized signatures into chunks and carries the ciphertexts in a fixed-size list beside the unchanged onion. Finally, the network cannot upgrade at once, so PQLN nodes have to coexist with unmodified nodes. The Lightning protocol lets unmodified nodes skip optional fields, so PQLN carries these additions in such fields. PQLN nodes can thus join today's network without a coordinated upgrade.

Our evaluation with real Lightning nodes shows that 1) the added cryptographic operations take at most 0.33~ms each and a PQ payment adds less than 60~ms per hop on a typical Internet link; 2) the main cost is communication, since a PQLN node downloads about 10 times as much gossip data as an unmodified node, or about 4 times as much with FN-DSA, the smaller Falcon-based signature scheme under standardization at NIST~\cite{falcon2020}; and 3) upgraded and unmodified nodes interoperate across the main Lightning operations, from channel opening to invoice, offer, and multi-hop payments.

Our contributions in this work are as follows:
\begin{itemize}
\item We propose PQLN, a PQ design for Lightning's off-chain surfaces that requires no change to Bitcoin and interoperates with unmodified nodes.
\item We turn LN's gossip into quantum-safe key distribution with trust-on-first-use pinning, and resist downgrade attacks by relying on pinned or offer-committed keys rather than forgeable negotiation messages.
\item We introduce signature chunking for the 639-byte invoice fields and a fixed-size ciphertext list beside the 1300-byte payment onion, which fit the large PQ material into message formats with hard size limits.
\item We analyze the security of PQLN against a quantum adversary and reduce each guarantee to the standard security notions of ML-DSA and ML-KEM.
\item We implement PQLN by modifying the source code of rust-lightning~\cite{ldk} and evaluate its computational and communication overhead, payment latency, and interoperability with real Lightning nodes, at every NIST parameter set and with FN-DSA in place of ML-DSA. Our code is publicly available at our GitHub repository.\footnote{\url{https://github.com/ahmet-kurt/pq-rust-lightning}}
\end{itemize}

The rest of the paper is organized as follows. Section~\ref{sec:related_work} reviews the related work. Section~\ref{sec:preliminaries} provides background on LN and the PQ standards. Section~\ref{sec:pqln_design} presents the design of PQLN and its threat model. Section~\ref{sec:security_analysis} analyzes its security. Section~\ref{sec:evaluation} presents our evaluation. Section~\ref{sec:discussion_and_limitations} discusses the limitations and Section~\ref{sec:conclusion} concludes the paper.

\section{Related Work}
\label{sec:related_work}

We review the efforts closest to PQLN and show why none of them closes this gap. Table~\ref{tab:related} summarizes the comparison, where the LN column denotes protecting Lightning's off-chain surfaces and the last column denotes running alongside unmodified nodes. Research on quantum attacks against Bitcoin has concentrated on the base layer~\cite{chevignard2026, babbush2026}, and so have the proposed mitigations. Stewart et al.~\cite{stewart2018} introduced a commit-delay-reveal protocol that migrates funds to a quantum-resistant scheme even after ECDSA is compromised. Bitcoin Improvement Proposal (BIP) 360 defines a Taproot-like output type that removes the quantum-vulnerable key-path spend~\cite{bip360}, and Kudinov and Nick~\cite{kudinov2025} tailored hash-based signatures to Bitcoin with much smaller sizes than the standardized Stateless Hash-Based Digital Signature Algorithm (SLH-DSA). Levy~\cite{qsb2026} proposed quantum-safe transactions that need no consensus change, though the scheme does not yet cover Lightning channels. All of these efforts stop at the base layer, and none attempt a Lightning design.

\begin{table}[!ht]
\caption{Comparison of PQLN with the Closest Related Work (\checkmark{} Yes, $\sim$ Partial, $\times$ No, n/a Not Applicable)}
\label{tab:related}
\centering
\scriptsize
\setlength{\tabcolsep}{3pt}
\begin{tabular}{@{}l c c c c@{}}
\hline
\textbf{Related work} & \textbf{LN} & \textbf{Design} & \textbf{Impl.} & \textbf{Interop.} \\
\hline
Bitcoin base-layer migration~\cite{bip360, stewart2018, kudinov2025, qsb2026} & $\times$ & \checkmark & $\sim$ & n/a \\
PQ transport~\cite{tlshybrid, kemtls2020, pqxdh, pqnoise2022} & $\times$ & \checkmark & \checkmark & n/a \\
PQ onion routing~\cite{ghoshkate2015, outfox2025} & $\times$ & \checkmark & $\sim$ & n/a \\
PQ adaptor signatures~\cite{esgin2020, tairi2021} & $\times$ & \checkmark & $\sim$ & $\times$ \\
Lightning PQ analysis~\cite{roasbeef2026pq} & $\sim$ & $\times$ & $\times$ & n/a \\
\textbf{PQLN (this work)} & \checkmark & \checkmark & \checkmark & \checkmark \\
\hline
\end{tabular}
\end{table}

Deployed protocols are further along in the same transition. TLS 1.3 negotiates hybrid groups that combine ML-KEM with an elliptic-curve exchange~\cite{tlshybrid}, and KEMs can even replace the handshake signatures~\cite{kemtls2020}. OpenSSH has defaulted to a hybrid PQ key exchange since version 9.0\footnote{\url{https://www.openssh.org/pq.html}} and Signal deployed a hybrid PQ handshake~\cite{pqxdh}. Angel et al.~\cite{pqnoise2022} relate most directly to our transport work, since they replaced the Diffie-Hellman (DH) operations of Noise with KEMs. None of these works targets Lightning or any payment channel network. PQLN brings the same migration practice to Basis of Lightning Technology (BOLT) 8 by hybridizing Noise\_XK while keeping the classical handshake byte-identical.

For onion routing, Ghosh and Kate~\cite{ghoshkate2015} proposed HybridOR, a hybrid circuit-extension key exchange for Tor. Rial et al.~\cite{outfox2025} later rebuilt the Sphinx packet format around KEMs with Outfox, which carries one KEM ciphertext per hop and therefore suits only fixed-length routes. Neither work targets Lightning's payment onion, whose packet carries a fixed 1300-byte payload over a variable number of hops. PQLN protects this surface by keeping the packet untouched and carrying the ML-KEM ciphertexts beside it (Section~\ref{sec:onion}).

PQ adaptor signatures target payment channels directly. Esgin et al.~\cite{esgin2020} built payment channel networks from the first PQ adaptor signature, and Tairi et al.~\cite{tairi2021} followed with an isogeny-based construction for privacy-preserving off-chain payments. Both would protect the funds inside channels once Bitcoin can verify their signatures, which it cannot do today, whereas PQLN protects the surfaces that no base-layer upgrade touches.

Kiayias and Litos~\cite{kiayias2020} formally proved the security of the Lightning protocol in the universal composition setting, and Kurt et al.~\cite{lngate2} extended Lightning to resource-constrained IoT devices through threshold cryptography, though both assume a classical adversary. The only Lightning-specific examination of the quantum threat is a forum post by Osuntokun~\cite{roasbeef2026pq}, which weighs PQ candidates for each layer and stops short of a design or an implementation. To the best of our knowledge, PQLN is thus the first design, implementation, and evaluation of PQ protection for Lightning's off-chain surfaces.

\section{Preliminaries}
\label{sec:preliminaries}

This section provides background on LN and the PQ standards. We explain the LN mechanisms on the example of Fig.~\ref{fig:overview}, in which Alice pays Carol through Bob, and note the elliptic-curve operation behind each mechanism.

\begin{table*}[!t]
\caption{Lightning Protocol Surfaces Vulnerable to Shor's Algorithm and Their Treatment in PQLN}
\label{tab:surfaces}
\centering
\scriptsize
\setlength{\tabcolsep}{4pt}
\begin{tabular}{c l l l}
\hline
\textbf{BOLT} & \textbf{Scope of the specification} & \textbf{Shor-vulnerable surface} & \textbf{Treatment in PQLN} \\
\hline
1, 10 & Base protocol and DNS bootstrap & None & Not applicable \\
2 & Peer protocol for channel management & ECDSA over on-chain transactions & Out of scope (rooted on-chain) \\
3 & Bitcoin transaction and script formats & secp256k1 keys inside Bitcoin Script & Out of scope (requires a Bitcoin change) \\
4 & Onion routing protocol & Per-hop ECDH & Protected, hybrid ML-KEM (Sec.~\ref{sec:onion}) \\
5 & On-chain transaction handling & ECDSA over on-chain transactions & Out of scope (rooted on-chain) \\
7 & P2P node and channel discovery & Node-key ECDSA on gossip messages & Protected, ML-DSA (Sec.~\ref{sec:gossip}) \\
8 & Encrypted and authenticated transport & ECDH in the Noise handshake & Protected, hybrid ML-KEM (Sec.~\ref{sec:transport}) \\
9 & Assigned feature flags & None & Two experimental feature bits registered \\
11 & Invoice protocol for payments & Node-key recoverable ECDSA & Protected, ML-DSA (Sec.~\ref{sec:invoices}) \\
12 & Flexible protocol for payments (offers) & Schnorr signatures and per-hop ECDH & Protected, ML-DSA and hybrid ML-KEM (Sec.~\ref{sec:offers}) \\
\hline
\end{tabular}
\end{table*}

\subsection{Lightning Network}
\label{sec:ln_prelim}

LN is a payment channel network on top of Bitcoin~\cite{poon2016}. Every LN node is identified by a secp256k1 public key called the \textit{node id}, with which it authenticates its transport sessions and signs its gossip messages and payment requests. The BOLT specifications define the protocol in numbered documents~\cite{bolts}.

\textit{Payment Channels:} Alice opens a channel to Bob with a \textit{funding transaction} that locks her coins into an output under the control of both parties, and the channel balances live in \textit{commitment transactions} that the parties update and sign at every payment. Broadcasting the latest commitment transaction closes the channel, while broadcasting an outdated one lets the counterparty claim the entire balance through a \textit{penalty transaction}. All three are ordinary Bitcoin transactions, so the safety of the funds inside a channel rests on Bitcoin's own cryptography.

\begin{figure}[!ht]
\centering
\includegraphics[width=0.9\columnwidth]{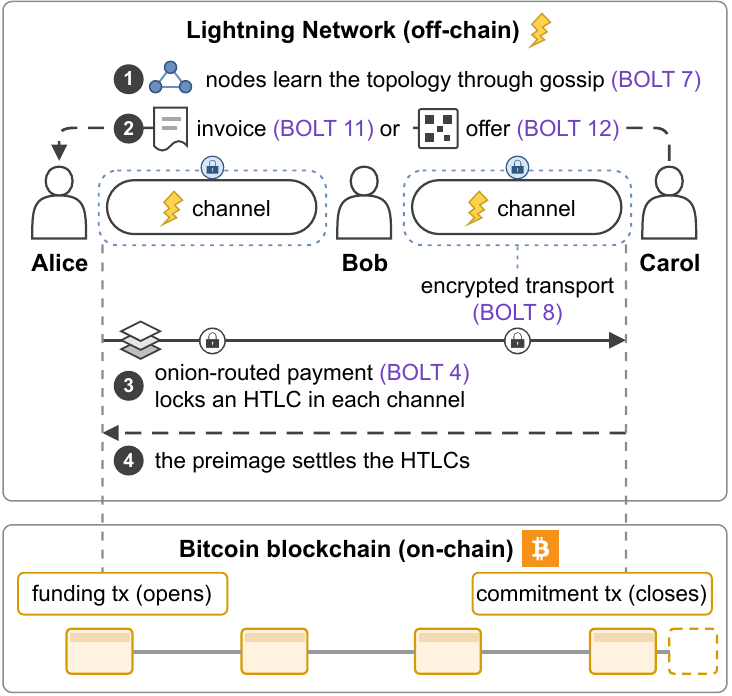}
\caption{Overview of an LN payment.}
\label{fig:overview}
\end{figure}

\textit{Gossip:} Nodes learn the network topology through the gossip messages of BOLT 7. Every node broadcasts a \textit{node\_announcement} that advertises its identity, features, and network addresses. A public channel enters the graph through a \textit{channel\_announcement} signed by both node keys and both funding keys, and both endpoints broadcast \textit{channel\_update} messages that advertise their forwarding fees and policies. All of these messages are signed with the announcing nodes' ECDSA keys, and nodes assemble them into the network graph to compute payment routes.

\textit{Transport:} Two connected nodes exchange all protocol messages inside an encrypted and authenticated session defined in BOLT 8. The session keys are derived with the Noise\_XK handshake~\cite{noise2018}, which runs three ECDH operations on secp256k1. The handshake authenticates the responder against its static public key, so the initiator must know this key in advance.

\textit{Payments:} A payment starts when Carol hands Alice a BOLT 11 \textit{invoice}, typically as a QR code. The invoice carries the amount and a \textit{payment hash} chosen by Carol, and she signs it with a recoverable ECDSA signature under her node key, from which Alice recovers Carol's node id. BOLT 12 \textit{offers} extend this flow with a small long-lived payment code that is published once and paid many times. A payer that scans an offer requests a fresh invoice from the payee, who signs it with a Schnorr signature. A payment can traverse multiple channels through intermediaries such as Bob, so Alice needs no channel to Carol. Every channel on the route locks the payment amount in a Hash Time Locked Contract (HTLC) that pays out against the \textit{preimage} of the payment hash and refunds the sender after a timeout. Carol therefore claims the payment by revealing the preimage, which settles every HTLC back along the route.

\textit{Payment Onion:} A payment's routing information travels in a Sphinx onion packet~\cite{sphinx2009} with a fixed payload of 1300 bytes, as defined in BOLT 4. The sender derives one ECDH secret per hop and encrypts the routing information in layers, so each hop learns only its predecessor and its successor. A payee that wants to hide its identity can additionally use \textit{blinded paths}, which replace the final hops of the route with encrypted routing data that it prepares with the same per-hop ECDH. The same construction also carries the \textit{onion messages} that BOLT 12 uses to request invoices.

\textit{Extensibility:} LN messages are extended with appended Type-Length-Value (TLV) records. BOLT 1 makes a record with an odd type optional, so an implementation that does not recognize it skips the record, whereas an unknown even type fails the message. The specifications call this rule \textit{it's ok to be odd}. BOLT 11 decoders likewise skip unknown tagged fields, and nodes advertise capabilities through the feature bits of BOLT 9. These mechanisms let upgraded and unmodified nodes coexist, and PQLN uses them to interoperate with today's network.

\subsection{Post-Quantum Standards}
\label{sec:pqc_prelim}

PQLN uses two of the PQ standards that NIST published in August 2024, namely ML-DSA for signatures and ML-KEM for key exchange. Both rest on lattice problems that no known quantum algorithm solves efficiently.

\textit{ML-DSA:} ML-DSA is standardized in FIPS 204 and originates from CRYSTALS-Dilithium~\cite{fips204}. It offers the same interface as ECDSA and Schnorr, so a protocol can attach an ML-DSA signature wherever it already produces a classical one. ML-DSA also binds an application-chosen \textit{context string} into every signature, so a signature created for one purpose cannot be reused for another.

\textit{ML-KEM:} ML-KEM is standardized in FIPS 203, originates from CRYSTALS-Kyber, and takes the place of DH key agreement~\cite{fips203}. Key generation produces an \textit{encapsulation key} and a \textit{decapsulation key}. Anyone holding the encapsulation key can \textit{encapsulate}, which yields a fresh shared secret and a \textit{ciphertext}, and decapsulating the ciphertext with the decapsulation key recovers the same secret. A protocol must therefore know a party's encapsulation key in advance and deliver the ciphertext to that party. Decapsulating a mismatched or tampered ciphertext does not fail but returns an unrelated secret, a behavior called \textit{implicit rejection}.

\textit{FN-DSA:} NIST also selected Falcon~\cite{falcon2020} for standardization in FIPS 206 and renamed it the FFT over NTRU-Lattice-Based Digital Signature Algorithm (FN-DSA), but had not yet published the standard at the time of writing. Its keys and signatures are much smaller than those of ML-DSA, but its key generation is far slower, and Section~\ref{sec:evaluation} measures it as an alternative to ML-DSA.

\section{PQLN Design}
\label{sec:pqln_design}

This section presents the design of PQLN. We first give an overview and the threat model, and then treat each surface in turn. Two hard constraints shape the design, one imposed by Bitcoin and one by the network itself. First, no Lightning upgrade can change Bitcoin's consensus rules, so the secp256k1 keys inside funding, commitment and penalty transactions stay as they are. Second, the network cannot be upgraded at once, so every mechanism has to interoperate with unmodified nodes, which we call \textit{vanilla} nodes. PQLN therefore targets the protocol surfaces that live purely off-chain and adds its protection where the extensibility rules of Section~\ref{sec:ln_prelim} let vanilla nodes skip it as unknown data. Table~\ref{tab:surfaces} maps every BOLT~\cite{bolts} to its Shor-vulnerable surface and to its treatment in PQLN. The surfaces of BOLTs 2, 3 and 5 are rooted in Bitcoin transactions and are therefore out of scope. The remaining five BOLTs carry cryptography that lives purely off-chain, and PQLN protects all five of them.

PQLN follows a \textit{hybrid conservative-extension} approach. It keeps all of the existing secp256k1 cryptography in place and adds a PQ primitive to every protected surface, as TLS and Signal do~\cite{tlshybrid, pqxdh}. Against a classical adversary, PQLN is therefore at least as secure as vanilla Lightning, and a quantum adversary additionally has to break the PQ primitive.

We use the two NIST PQ standards of Section~\ref{sec:pqc_prelim}. Signatures use ML-DSA-44 with public keys of 1312 bytes and signatures of 2420 bytes~\cite{fips204}. We chose the smallest standardized parameter set because signature and key material dominates our wire overhead. FN-DSA would shrink this material further, but its standard is not final, so we keep ML-DSA-44 as the default and measure FN-DSA as an alternative in Section~\ref{sec:evaluation}. Key exchange uses ML-KEM-768, the set that TLS and OpenSSH adopted, with encapsulation keys of 1184 bytes, ciphertexts of 1088 bytes and 32-byte shared secrets~\cite{fips203}. Every node derives its ML-DSA signing key and its static ML-KEM key from its existing seed at two dedicated hardened BIP 32 indices~\cite{bip32}, so the wallet backup that restores its classical identity also restores its PQ identity. Hardened derivation is one-way, so a node key recovered with Shor's algorithm reveals nothing about the PQ keys derived from the same seed.

The gossip layer distributes these keys and serves as the trust anchor of PQLN. Every node publishes its ML-DSA and ML-KEM public keys inside its signed \textit{node\_announcement}, and every PQLN node pins these keys on first sight. All other surfaces then verify against the pinned keys under this \textit{trust-on-first-use} model. The only exception is BOLT 12, where an offer additionally commits a per-offer signing key so that invoice verification works even for payees that never appeared in gossip. PQLN also registers two experimental feature bits so that peers can discover support. However, a quantum adversary can flip these bits, so no security decision depends on them.

The signature surfaces of gossip, invoices and offers are always on, so a PQLN node attaches its PQ signatures wherever it signs. A vanilla peer ignores them, whereas a verifier with a trusted key gains the protection. The remaining surfaces, namely the transport, the payment onion and the blinded paths, rely on key exchange and need both endpoints to take part. They are therefore explicit opt-ins through a dedicated PQ port and the three configuration flags of Sections~\ref{sec:offers} and~\ref{sec:onion}.

\subsection{Threat Model}
\label{sec:threat}

We assume an adversary that possesses a cryptographically relevant quantum computer and runs Shor's algorithm at scale~\cite{shor1997}. It recovers the secret key behind any secp256k1 public key that it observes~\cite{aggarwal2018}. It can therefore forge the ECDSA and Schnorr signatures of any node whose key appeared in gossip, in an invoice or on the wire, and it can compute the shared secret of any ECDH exchange from the two public keys alone. Symmetric primitives and hash functions only face Grover's quadratic speedup, and the parameter sizes in use absorb it with an ample margin~\cite{grover1996}. We therefore treat ChaCha20-Poly1305, SHA-256 and HMAC as secure. The adversary attacks on two timelines. In the retroactive timeline, it records classical traffic today and breaks it once a quantum computer exists. This \textit{harvest-now-decrypt-later} strategy already endangers Lightning's confidentiality surfaces~\cite{mosca2018, nistir8547}. In the live timeline, it attacks in real time by forging signatures, substituting keys and inserting itself into new sessions.

We also grant the adversary full control of the network, so it observes, delays, modifies and injects messages between any pair of nodes. Moreover, it participates in the network itself by running well-connected nodes, opening channels, relaying gossip and forwarding payments. It can also act as the counterparty, the routing hop or the payee of a victim. Based on this adversary, we consider the following attacks against Lightning:

\begin{itemize}
\item \textit{Threat 1. Node Impersonation:} The adversary forges a victim node's gossip signatures to announce substituted keys or attacker-chosen parameters on the victim's behalf.
\item \textit{Threat 2. Transport Decryption:} The adversary breaks the ECDH operations of the BOLT 8 handshake to decrypt recorded sessions or to impersonate a responder in new sessions.
\item \textit{Threat 3. Invoice Forgery:} The adversary forges the payee signature on a BOLT 11 or BOLT 12 invoice to substitute the payment hash, the amount or the payment paths.
\item \textit{Threat 4. Payment Deanonymization:} The adversary recovers the per-hop ECDH secrets of the payment onion, unwraps its layers and links payer, route and payee. The same capability strips the privacy of blinded paths and onion messages.
\item \textit{Threat 5. Downgrade Attacks:} The adversary strips or tampers with the PQ additions of a message so that the receiver processes it as a vanilla message, and then breaks the exchange like any other vanilla exchange.
\end{itemize}

Several threats are outside our model. Lightning-level mechanisms cannot prevent attacks on the Bitcoin layer, which await Bitcoin's own PQ transition~\cite{aggarwal2018}. These include the theft of channel funds by breaking the secp256k1 keys of funding, commitment or HTLC outputs. Denial-of-service attacks are orthogonal to quantum resistance. We also assume that endpoints are not compromised. Finally, the trust-on-first-use model assumes that a node pinned the keys of its peers before a cryptographically relevant quantum computer exists, since a first contact made afterwards can be intercepted by a live adversary. Section~\ref{sec:security_analysis} states this assumption formally and Section~\ref{sec:discussion_and_limitations} discusses it.

\subsection{PQ Gossip (BOLT 7)}
\label{sec:gossip}

We begin with gossip because it is the trust anchor of PQLN. The \textit{node\_announcement} and \textit{channel\_update} messages of Section~\ref{sec:ln_prelim} are signed only with the announcing node's ECDSA key, so a quantum adversary could rewrite them freely (Threat 1). Every node already receives a \textit{node\_announcement} from every other node, so PQLN uses this message to publish a node's ML-DSA and ML-KEM public keys to the whole network.

\begin{figure}[!t]
\centering
\includegraphics[width=0.8\columnwidth]{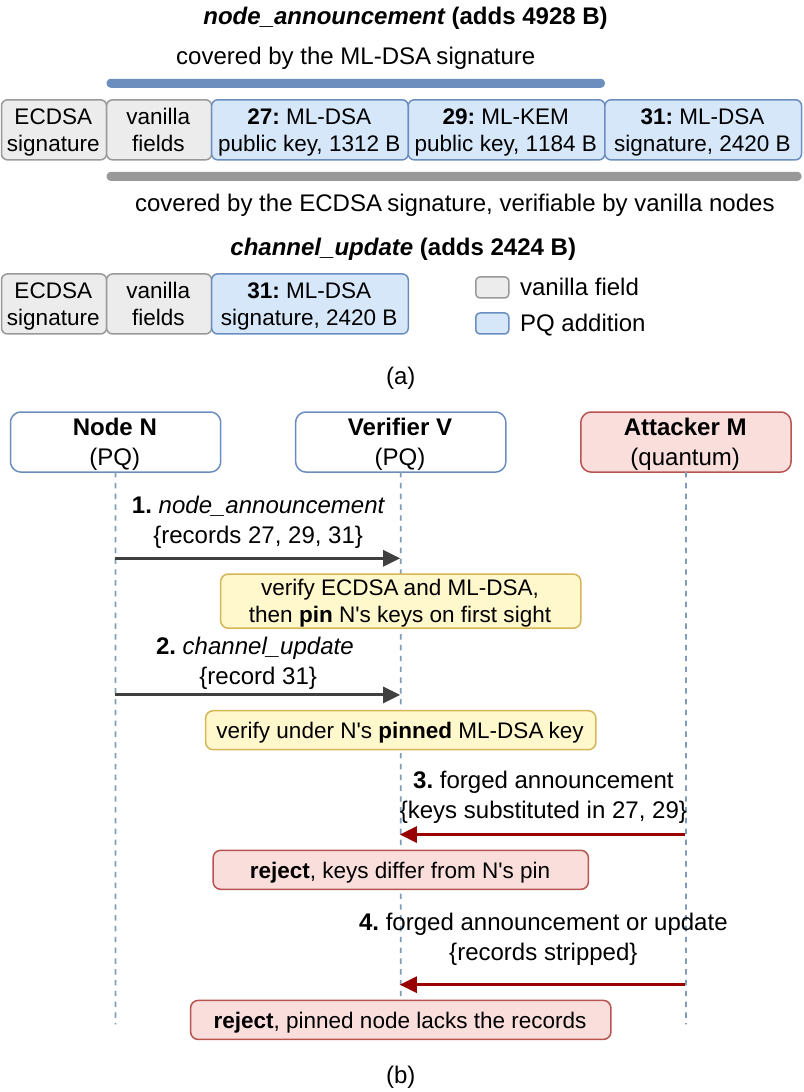}
\caption{The PQLN gossip design. (a) The extended gossip messages. The ML-DSA signature covers the published keys and the ECDSA signature covers the complete message. (b) The verifier pins a node's keys on first sight and afterwards rejects substituted keys and stripped records.}
\label{fig:gossip}
\end{figure}

Fig.~\ref{fig:gossip}(a) shows the modified messages. A PQLN node appends three odd TLV records to the tail of its \textit{node\_announcement}. Record 27 carries the node's ML-DSA public key, record 29 carries its static ML-KEM encapsulation key and record 31 carries an ML-DSA signature over the serialized announcement, including the two key records. The ML-DSA signature is computed under a gossip-specific domain-separation context and appended last, and the node then produces its classical ECDSA signature over the complete message. The classical signature thus commits to the PQ records, so a vanilla node verifies the announcement exactly as it does today and skips the unknown odd records under the \textit{it's ok to be odd} rule. A \textit{channel\_update} is extended the same way but carries only the signature record, since the verifier already knows the signer's ML-DSA key from its \textit{node\_announcement}.

Fig.~\ref{fig:gossip}(b) shows how a PQLN node processes incoming gossip. The verifier runs the PQ checks after the vanilla checks and updates its pins only after every check has passed. A rejected message therefore never alters a pin. Specifically, the verifier first checks the classical ECDSA signature as usual and then parses the PQ records. If the announcement carries a key without the signature, or the signature without the ML-DSA key, the verifier rejects it as malformed. Otherwise, the verifier checks the ML-DSA signature against the embedded key. If the node is new, the verifier pins its ML-DSA and ML-KEM keys. If the node is already pinned, the verifier rejects the announcement in three cases, namely when a key differs from the pinned key, when a pinned key is missing, or when the signature is missing. A \textit{channel\_update} from a pinned node must likewise carry a valid ML-DSA signature under the pinned key, but it carries no keys and therefore never establishes a pin.

Gossip relay needs one adjustment. Vanilla rust-lightning refuses to relay a gossip message when its unrecognized trailing data exceeds 1024 bytes, so a vanilla node accepts, verifies and stores a PQ announcement but does not forward it. PQLN raises this relay budget to 8192 bytes, so PQ gossip propagates across the PQ-aware part of the network, while a vanilla node on the way stops the propagation without rejecting the message. The budget caps the unrecognized data that a node relays per message at 1~kB for a vanilla node and 8~kB for a PQLN node, and relaying still requires valid signatures and a channel with an on-chain funding output, so the extra bytes stay tied to funded channels. Section~\ref{sec:interop} confirms this behavior with real nodes.

We deliberately leave \textit{channel\_announcement}, the third signed gossip message, classical. Two of its four signatures prove ownership of the on-chain funding output and cannot be protected off-chain, and protecting the two node-key signatures alone would leave a message that a quantum adversary can still half forge. A node's identity keys and forwarding parameters stay fully protected through \textit{node\_announcement} and \textit{channel\_update}.

\subsection{PQ Transport (BOLT 8)}
\label{sec:transport}

Every Lightning message travels inside the BOLT 8 transport of Section~\ref{sec:ln_prelim}. Its Noise\_XK handshake derives the session keys from three ECDH operations on secp256k1. The operation against the responder's static key also authenticates the responder, and Shor's algorithm breaks all three operations. A recorded session therefore becomes decryptable retroactively, and a live quantum adversary can impersonate any responder after observing its static key (Threat 2).

PQLN hybridizes the handshake with two ML-KEM encapsulations, as shown in Fig.~\ref{fig:transport}. The initiator must know the responder's static ML-KEM key in advance, from the gossip pin of Section~\ref{sec:gossip} or out of band, just as Noise\_XK already assumes for the classical static key. In act one, the initiator encapsulates to that pinned key and appends the ciphertext $ct_s$ together with a freshly generated ephemeral ML-KEM public key $ek_e$. Only the true responder can decapsulate $ct_s$, which authenticates it against a quantum impersonator. The responder in turn encapsulates to $ek_e$ and appends the resulting ciphertext $ct_e$ to act two. This provides forward secrecy, since a later compromise of the static keys does not reveal the ephemeral secret. Act three stays unchanged. Each of the first two acts mixes its ECDH secret into the Noise chaining key, and the hybrid handshake folds the ML-KEM shared secret of the same act into the chaining key right afterwards. Both ciphertexts and $ek_e$ are also absorbed into the transcript hash, so the session keys become hybrid over five shared secrets.

\begin{figure}[!t]
\centering
\includegraphics[width=0.8\columnwidth]{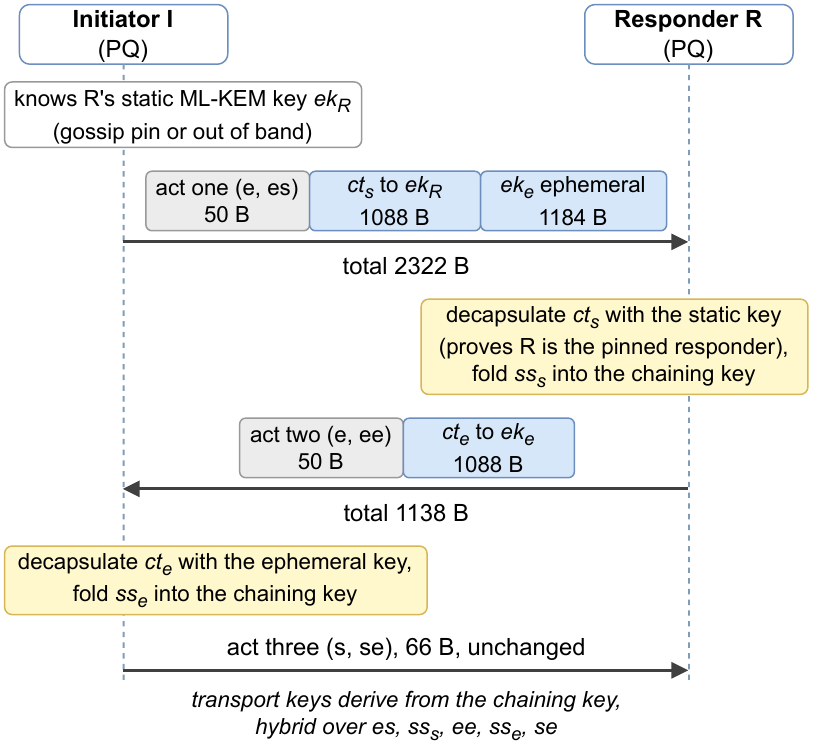}
\caption{The PQLN BOLT 8 handshake.}
\label{fig:transport}
\end{figure}

PQLN deliberately has no in-band negotiation of PQ support, because a negotiation message is exactly what a quantum adversary could rewrite to force a downgrade (Threat 5). The hybrid handshake instead runs on a dedicated port, and the classical port stays byte-identical to vanilla. In our implementation, the operator supplies the PQ address of a peer, and a deployment would advertise it as an additional address in the \textit{node\_announcement}, which the pinned ML-DSA signature already covers. Act three still authenticates the initiator classically, so a quantum adversary could connect under another node's identity. This is of little value, since Lightning nodes accept inbound connections from anonymous peers by design and the other surfaces of PQLN authenticate every further action of a connected peer.

\subsection{PQ Invoices (BOLT 11)}
\label{sec:invoices}

A BOLT 11 invoice commits to the payment hash, the amount and the destination of a payment. Its recoverable ECDSA signature is all that binds the request to the payee, and the payee's node id is usually recovered from that signature rather than carried in the invoice. A quantum adversary that forges the signature can therefore substitute the payment hash and the destination of an intercepted invoice and collect the payment (Threat 3).

PQLN adds the payee's ML-DSA signature and optionally the corresponding public key to the invoice. A BOLT 11 tagged field carries at most 639 bytes because of its 10-bit length encoding. PQLN therefore splits the two values into chunks and carries them as several tagged fields under two unassigned tag values, three fields for the public key and four for the signature (Fig.~\ref{fig:invoice}, top). A vanilla decoder skips unknown tagged fields, so the invoice still parses everywhere. The ML-DSA signature is computed over the human-readable part and the data part of the invoice, including the embedded public key fields, under an invoice-specific domain-separation context. The classical signature is produced afterward so that it commits to the PQ fields. Embedding the public key makes the invoice self-contained, while a payee that expects to be found in gossip can omit it and keep the invoice small enough for a QR code.

\begin{figure}[!t]
\centering
\includegraphics[width=0.8\columnwidth]{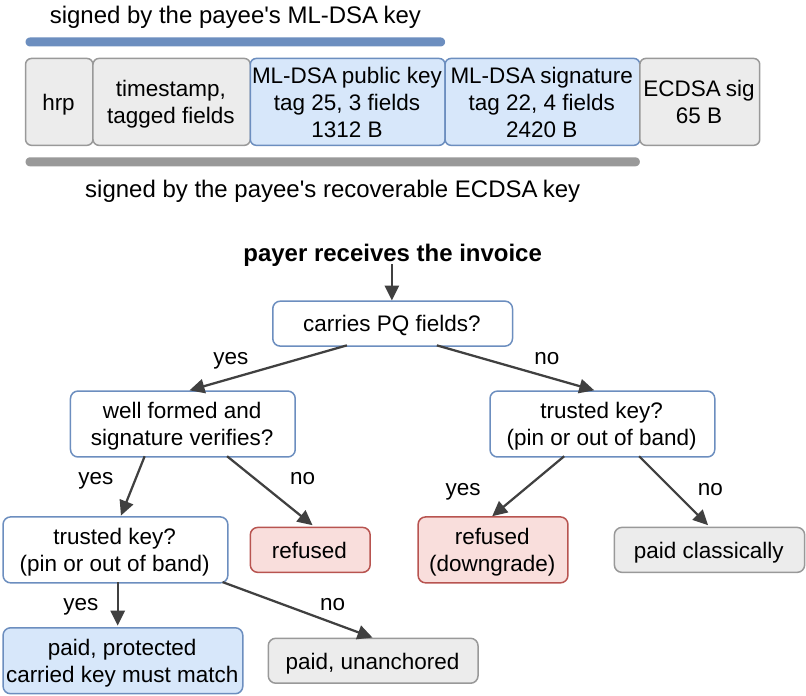}
\caption{The PQLN BOLT 11 invoices.}
\label{fig:invoice}
\end{figure}

Verification runs inside the pay path with no opt-in and applies the following policy (Fig.~\ref{fig:invoice}, bottom). Before dispatching any HTLC, the payer resolves a trusted key for the payee, from an out-of-band source or from the payee's gossip pin. When a trusted key exists, the payer requires a valid ML-DSA signature under that key and refuses the invoice as a downgrade if the signature is missing. It also refuses the invoice if an embedded public key differs from the trusted key. Without a trusted key, a vanilla invoice is paid as today, and an invoice that carries PQ fields is paid only if they are well formed and the signature verifies under any carried key. We call this last case \textit{unanchored}, and it gives no protection because a signature under a self-asserted key proves nothing against an adversary who can mint both. A first contact with an unannounced payee therefore stays unprotected.

\begin{figure}[!t]
\centering
\includegraphics[width=0.8\columnwidth]{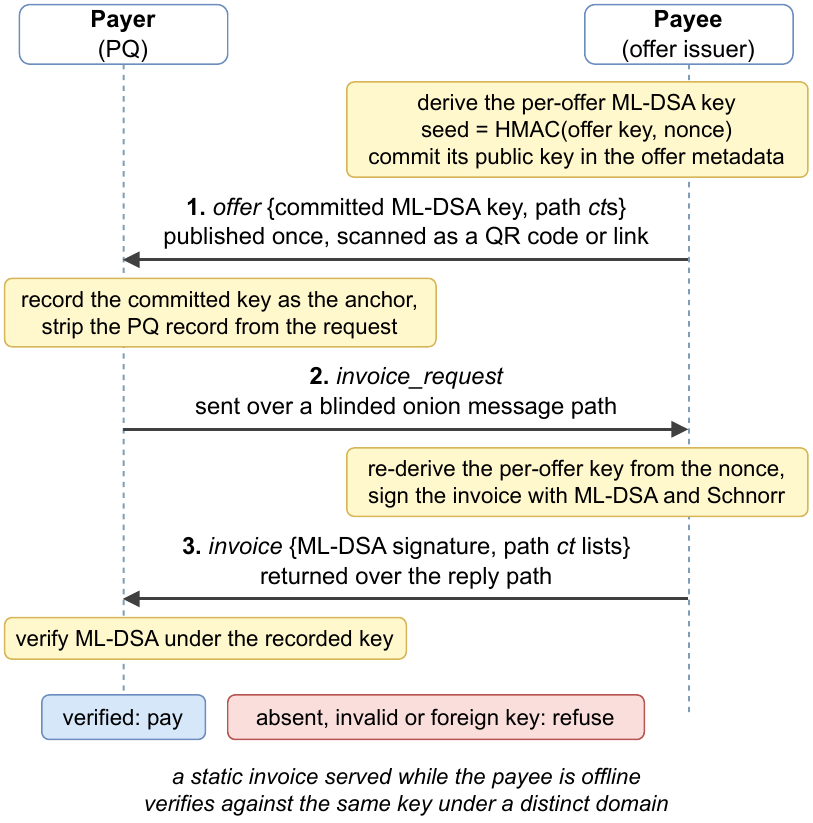}
\caption{The PQLN BOLT 12 offers.}
\label{fig:offers}
\end{figure}

\begin{figure*}[!t]
\centering
\includegraphics[width=0.85\textwidth]{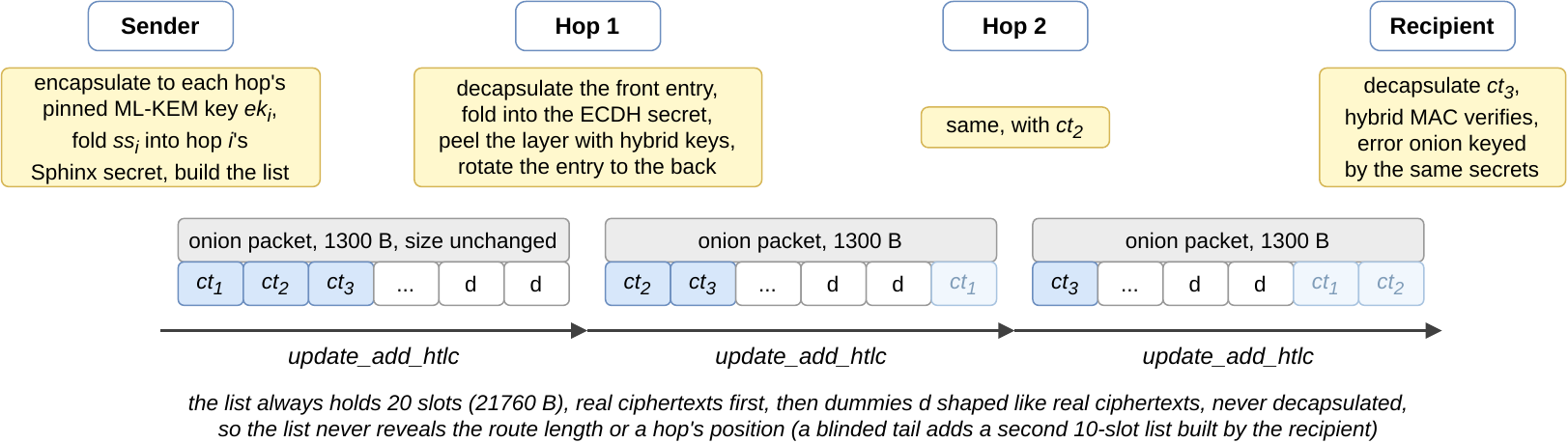}
\caption{The PQLN payment onion.}
\label{fig:onion}
\end{figure*}

\subsection{PQ Offers (BOLT 12)}
\label{sec:offers}

With BOLT 12 offers, the payer sends an \textit{invoice\_request} to the payee over a blinded onion message path and receives a freshly signed invoice in return. The returned invoice is bound to the offer only by a Schnorr signature, so a quantum adversary between the two parties can forge it and substitute the payment hash, the amount or the payment paths (Threat 3). The gossip pin rarely helps here, since offer payees are commonly unannounced nodes reachable only through blinded paths. The privacy of the exchange in turn rests on the per-hop ECDH of those paths (Threat 4).

PQLN closes this gap with a trust anchor already in the payer's possession, namely the offer itself. When a PQLN node builds an offer, it derives a per-offer ML-DSA key and commits the public key inside the offer, as shown in Fig.~\ref{fig:offers}. The signing seed is an HMAC of the node's symmetric offer key and the offer's nonce, so a quantum adversary cannot recover it from any published key, and distinct offers carry unlinkable keys. The committed key is carried in the offer's metadata record, together with one ML-KEM ciphertext per PQ blinded path. Vanilla payers copy this record verbatim into the \textit{invoice\_request} and payees echo it back, so every classical implementation handles it consistently, whereas a PQLN payer strips it from its request.

The payee re-derives the per-offer key from the offer's nonce and signs the responding invoice under a BOLT 12 domain-separation context. The signature is placed in an odd record of the invoice's experimental TLV range and written into the unsigned invoice, so the classical Schnorr signature covers it. A vanilla node ignores the unknown record but keeps it in place and still verifies the classical signature as usual. Verification again runs inside the pay path with no opt-in. A payer that scanned a PQ offer records the committed key, and the received invoice must verify under exactly that key before any HTLC is dispatched. The same anchor covers the static invoices that an often-offline payee pre-signs for asynchronous payments, and these verify under a distinct signing domain, so they cannot be replayed as regular invoices.

The privacy surfaces of BOLT 12 use the hybrid key exchange of the payment onion. The recipient enables PQ blinded paths with the \textit{build\_\allowbreak post\_\allowbreak quantum\_\allowbreak blinded\_\allowbreak paths} flag. The offer's message paths, the reply path of the \textit{invoice\_request} and the blinded payment paths inside the invoice are then built with the hybrid route blinding of Section~\ref{sec:onion}. Their ciphertexts reach the senders in the offer metadata, beside the reply path in the onion message and in a second experimental invoice record, respectively. Refund paths are built the same way. Two BOLT 12 signatures remain classical, namely the payer's signature on the \textit{invoice\_request} and the payee's signature on a refund's invoice, since the verifier has no trusted key for either signer.

\subsection{PQ Payment Onion (BOLT 4)}
\label{sec:onion}

The Sphinx packet of Section~\ref{sec:ln_prelim} carries a fixed 1300-byte payload, and the sender builds it from one ECDH secret per hop. Each hop derives the keys of its own layer from its ECDH secret, so a quantum adversary that recovers these secrets unwraps the route hop by hop, links payer to payee and reads every per-hop payload (Threat 4). An ML-KEM ciphertext of 1088 bytes would nearly fill the packet, and vanilla nodes cannot forward a packet of any other size. PQLN therefore keeps the onion payload at 1300 bytes and instead makes the per-hop keys hybrid.

Fig.~\ref{fig:onion} shows the construction. For every hop, the sender encapsulates to the hop's pinned ML-KEM key and folds the shared secret into the hop's classical Sphinx secret, so the encryption and MAC keys of that layer become hybrid while the onion keeps its vanilla format. The ciphertexts travel beside the onion in a new odd field of \textit{update\_add\_htlc} as a fixed-size list of 20 slots, the onion's maximum hop count. Real entries come first in hop order and dummies fill the rest, so the list looks the same for every route length. The dummies are not random bytes but are compressed and encoded from uniformly random polynomial coefficients exactly as FIPS 203 builds a ciphertext, since the compression makes some encoded values more likely than others and a hop could otherwise count the real entries. Each hop decapsulates the front entry, folds the secret into its classical ECDH result, peels its layer with the hybrid keys and rotates the entry to the back. The list carries no key or MAC of its own, and its integrity comes from the hybrid onion MAC instead, so a tampered, reordered or dropped entry yields a wrong hybrid secret at that hop and the payment fails closed.

The hybrid secret also becomes the hop's incoming secret, so the return error onion and its attribution data inherit the protection. Blinded payment paths are built by the recipient instead of the sender, so the recipient protects them. When building a PQ blinded path, it encapsulates to each path hop's ML-KEM key and folds the secrets into the route-blinding schedule, so the blinded node ids and the encrypted per-hop data are hybrid before the sender sees the path. The ciphertexts reach the sender in the invoice record of Section~\ref{sec:offers} and travel in a second fixed list of 10 slots beside the onion.

A key exchange needs both endpoints to take part, so PQLN cannot force the hybrid onion on every route. Two configuration flags therefore decide when a node requires it. A sender builds the hybrid onion whenever every hop has a pinned ML-KEM key and otherwise falls back to a classical onion, unless the \textit{require\_\allowbreak post\_\allowbreak quantum\_\allowbreak payments} flag forbids the fallback. Similarly, the \textit{require\_\allowbreak post\_\allowbreak quantum\_\allowbreak inbound} flag makes a forwarding or receiving node fail back any HTLC that arrives without the hybrid protection (Threat 5).

\section{Security Analysis}
\label{sec:security_analysis}

In this section, we show how PQLN addresses the threats of Section~\ref{sec:threat}. We give one proposition per threat, each with a proof sketch that reduces its guarantee to the standard notions of ML-DSA and ML-KEM.

\subsection{Model and Assumptions}
\label{sec:sec_model}

We write $(K, c) \leftarrow \mathsf{Encaps}(ek)$ for ML-KEM-768 encapsulation, whose ciphertext $c$ decapsulates to $K$ under the decapsulation key $dk$, and $\mathsf{Vf}(pk, m, \sigma, \mathsf{ctx})$ for ML-DSA-44 verification under the per-surface context string $\mathsf{ctx}$. The hash $H$ is SHA-256. A node $N$ holds an ML-DSA key pair $(pk_N, sk_N)$ and a static ML-KEM key pair $(ek_N, dk_N)$. The pin $\mathsf{Pin}_V[N]$ denotes the pair $(pk_N, ek_N)$ that a verifier $V$ stored on first sight, or $\bot$ if $V$ never saw $N$. Outside the transport, PQLN derives the hybrid secret
\begin{equation}
s^{*} = H(\tau \,\|\, s_c \,\|\, K)
\label{eq:fold}
\end{equation}
under a per-surface tag $\tau$ wherever it folds a classical ECDH secret $s_c$ with an ML-KEM secret $K$. The transport folds $K$ into Noise's chaining key through HKDF instead.

We model the adversary $\mathcal{A}$ of Section~\ref{sec:threat} by giving it a discrete-logarithm oracle on secp256k1. It thus knows the secret key behind every public key and the shared secret of every ECDH exchange that it has observed. It also controls the network and runs nodes of its own, while the endpoints under analysis are honest.

\begin{assumption}
ML-DSA-44 is existentially unforgeable under chosen-message attacks (EUF-CMA) and ML-KEM-768 is indistinguishable under chosen-ciphertext attacks (IND-CCA) against quantum adversaries, as the standards claim~\cite{fips204, fips203}. ChaCha20-Poly1305 is a secure authenticated cipher, and $H$, HMAC and HKDF behave as random oracles or pseudorandom functions, as Section~\ref{sec:threat} assumes.
\end{assumption}

\begin{assumption}
For every honest node $N$ and verifier $V$ with $\mathsf{Pin}_V[N] \neq \bot$, the pin holds the keys that $N$ generated, and every offer that a payer holds carries the key that its payee committed. This trust-on-first-use assumption of Section~\ref{sec:threat} holds when the announcement or the offer arrived either before a quantum adversary existed or over an authenticated channel such as a QR code scanned in person.
\end{assumption}

Against a classical adversary, PQLN is at least as secure as vanilla Lightning, since every classical check stays in place and a message that must carry both signatures is unforgeable when either scheme is~\cite{bindel2017}. Moreover, (\ref{eq:fold}) follows the random-oracle KEM combiner of Giacon et al.~\cite{giacon2018} with the ciphertexts left out of the hash. This is safe because ML-KEM decapsulation binds the shared secret to its ciphertext through re-encryption~\cite{fips203} and a different ephemeral key yields a different ECDH secret. Cremers et al. formalize the former property as HON-BIND-K-CT, under which the output key of an honestly generated key pair determines the ciphertext, and report that the known attacks on ML-KEM's binding need a maliciously formed key pair~\cite{cremers2024}, which Assumption 2 excludes. The output of (\ref{eq:fold}) thus stays random when either input secret is, so the propositions below address only the quantum adversary.

\subsection{Security Against the Threats}
\label{sec:sec_props}

\begin{proposition}[Node impersonation]
Under Assumptions 1 and 2, if a verifier $V$ with $\mathsf{Pin}_V[N] \neq \bot$ accepts a \textit{node\_announcement} or \textit{channel\_update} attributed to an honest node $N$, then $N$ produced it, except with probability at most $\mathbf{Adv}^{\mathrm{EUF\text{-}CMA}}_{\mathrm{ML\text{-}DSA}}(\mathcal{B})$ for an adversary $\mathcal{B}$ of similar running time.
\end{proposition}

\begin{IEEEproof}
The verifier of Section~\ref{sec:gossip} accepts such a message only if every expected record is present, any carried key equals $\mathsf{Pin}_V[N]$ and the ML-DSA signature verifies under $pk_N$ over the entire message, including the key records. The discrete-logarithm oracle lets $\mathcal{A}$ forge the ECDSA signature and set any feature bit, but neither enters this decision. A stripped record fails the first check, and a substituted key fails the second because the pin holds the keys of $N$ by Assumption 2. What remains is a message that $N$ never signed but that verifies under $pk_N$ and the gossip context, which $\mathcal{B}$ outputs as its forgery.
\end{IEEEproof}

\begin{proposition}[Transport decryption]
Let an honest initiator $I$ hold $ek_R$ of an honest responder $R$ under Assumption 2. Under Assumption 1, 1) the session keys of their hybrid handshake are indistinguishable from random for $\mathcal{A}$ even given all three ECDH secrets, 2) no party without $dk_R$ completes act two toward $I$, and 3) a later compromise of both static keys does not reveal the session keys.
\end{proposition}

\begin{IEEEproof}
The chaining key absorbs the shared secrets $K_s$ and $K_e$ behind the ciphertexts $ct_s$ and $ct_e$ of Section~\ref{sec:transport} between the three ECDH secrets. Replacing $K_e$ by a random key changes the view of $\mathcal{A}$ by at most the IND-CCA advantage at $ek_e$. The chaining key has then absorbed a random secret that $\mathcal{A}$ does not hold, so the session keys that it yields are pseudorandom, which gives 1). This step rests on the hash-object result of Angel et al., who prove that the Noise hash chain yields pseudorandom keys once it absorbs one input unknown to the adversary, and their argument covers hybrid patterns~\cite{pqnoise2022}. We do not restate their proof in the flexible ACCE model, so the proposition argues at the level of the chaining key only. The act two tag that $I$ verifies is keyed after $K_s$ has entered the chaining key. A party without $dk_R$ decapsulates an unrelated secret by implicit rejection, so it completes act two only by recovering $K_s$ from $ct_s$ or by forging the tag, which gives 2). Finally, $K_e$ depends only on the per-connection ephemeral pair, so a later static compromise reveals nothing about it, which gives 3).
\end{IEEEproof}

\begin{proposition}[Invoice forgery]
Let a payer $P$ hold an anchor $pk^{*}$ for its payee under Assumption 2, namely $\mathsf{Pin}_P[N]$ or an out-of-band key for BOLT 11 and the key committed in the offer for BOLT 12. Under Assumption 1, $\mathcal{A}$ makes $P$ pay an invoice that the payee did not sign with probability at most $\mathbf{Adv}^{\mathrm{EUF\text{-}CMA}}_{\mathrm{ML\text{-}DSA}}(\mathcal{B})$.
\end{proposition}

\begin{IEEEproof}
The pay paths of Sections~\ref{sec:invoices} and~\ref{sec:offers} refuse an absent, malformed or invalid signature and a carried key that differs from the anchor. The payer therefore dispatches an HTLC only if $\mathsf{Vf}(pk^{*}, m_I, \sigma_I, \mathsf{ctx}) = 1$, where $m_I$ covers the payment hash, the amount, the payment paths and any embedded key. Such a signature on an invoice that the payee never signed is a forgery, and the pairwise distinct contexts prevent a signature from one surface, including a static invoice, from verifying on another. The per-offer key of BOLT 12 also stays hidden from $\mathcal{A}$, since its seed is an HMAC under the node's symmetric offer key.
\end{IEEEproof}

\begin{proposition}[Payment deanonymization]
Consider a hybrid onion over $n$ hops whose ML-KEM keys satisfy Assumption 2. Under Assumption 1, the per-hop keys of every hop are indistinguishable from random for $\mathcal{A}$ even given all per-hop ECDH secrets, with advantage at most $n \cdot \bigl(\mathbf{Adv}^{\mathrm{IND\text{-}CCA}}_{\mathrm{ML\text{-}KEM}}(\mathcal{B}) + q_H / 2^{256}\bigr)$ for $q_H$ random oracle queries.
\end{proposition}

\begin{IEEEproof}
Hop $i$ derives its keys from $s^{*}_i$ of (\ref{eq:fold}) with $(K_i, c_i) \leftarrow \mathsf{Encaps}(ek_i)$. A hybrid argument over the hops replaces each $K_i$ by a random key at the cost of one IND-CCA advantage. Each $s^{*}_i$ is then a random oracle output on an input that $\mathcal{A}$ guesses with probability at most $q_H / 2^{256}$ per hop. From that point, the Sphinx analysis of Danezis and Goldberg~\cite{sphinx2009} applies unchanged to the packet, and it covers blinded paths as well, since the route-blinding cascade applies (\ref{eq:fold}) before it derives the blinded node ids. The ciphertext list carries no key or MAC of its own and its dummies follow the ciphertext distribution, so a single hop learns nothing from it beyond the fixed-size onion as long as ML-KEM ciphertexts are pseudorandom and hide their target key, which Maram and Xagawa proved for Kyber~\cite{maram2023}. However, the list is rotated rather than re-randomized at each hop, so two colluding hops can match its contents and learn that they forward the same payment. Vanilla Lightning already leaks this linkage through the payment hash that every HTLC of a payment carries, so PQLN neither weakens nor improves the privacy of the onion against colluding hops. Finally, a tampered entry changes $K_i$ by implicit rejection and fails the MAC check.
\end{IEEEproof}

\begin{proposition}[Downgrade]
Under Assumptions 1 and 2, $\mathcal{A}$ cannot make an anchored verifier process a gossip message or an invoice classically except by forging ML-DSA, nor make a PQ port complete a classical handshake. Under the two flags of Section~\ref{sec:onion}, a sender never dispatches a classical onion and a node never accepts one.
\end{proposition}

\begin{IEEEproof}
The accept predicate of every signature surface depends only on the verifier's local anchor and on ML-DSA verification, never on a feature bit. Since $\mathcal{A}$ rewrites messages but not the verifier's storage, a message from an anchored peer is either verified under the anchor or refused, and Propositions 1 and 3 bound the chance that a forged message verifies. The transport offers no negotiation to rewrite, since the hybrid handshake runs on its own port with a fixed act layout. For the onion, the sender flag refuses any route with a hop that lacks a pinned key and the inbound flag fails back any HTLC without a ciphertext list, so a withheld gossip key or a stripped list stops the payment instead of degrading it.
\end{IEEEproof}

\section{Evaluation}
\label{sec:evaluation}

This section describes the implementation and experiment setup of PQLN and presents the evaluation results.

\begin{figure*}[!t]
\centering
\includegraphics[width=0.85\textwidth]{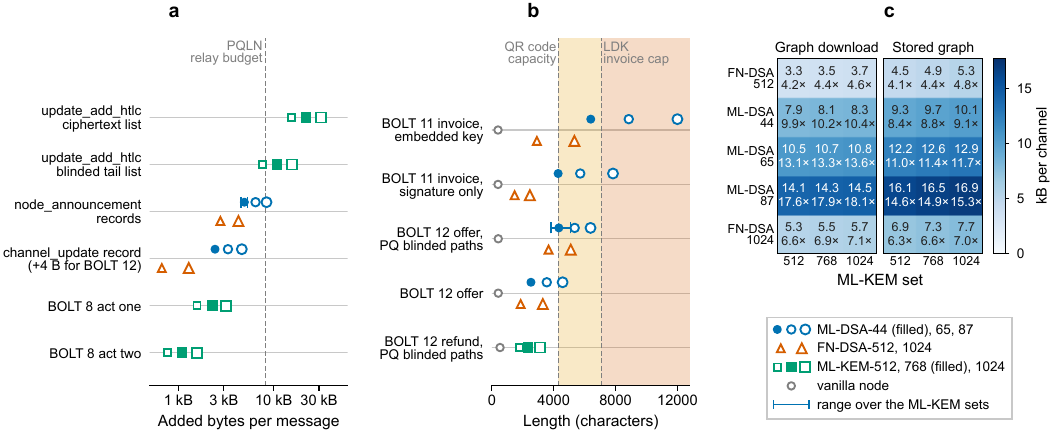}
\caption{Communication overhead of PQLN at every ML-DSA, FN-DSA, and ML-KEM parameter set. (a) Bytes added per message. (b) Payment request lengths. (c) Gossip cost per channel for all 15 pairings, where each cell also gives the multiple of the vanilla cost of 0.8~kB (download) and 1.1~kB (stored graph) per channel.}
\label{fig:paramsets}
\end{figure*}

\subsection{Experiment Setup and Metrics}
\label{sec:setup}

To evaluate the proposed design, we implemented it by modifying the rust-lightning~\cite{ldk} source code at commit 384e0d6 of its main branch (August 2026). The fork adds roughly 11,000 lines of code across 52 files, all gated behind a single Cargo feature named \textit{post-quantum}. ML-DSA-44 and ML-KEM-768 come from the fips204 (v0.4.6) and fips203 (v0.4.3) Rust crates. With the feature disabled, the fork passes the complete upstream test suite. With it enabled, the fork also passes 99 added tests, 36 of which exercise the attacks and refusals analyzed in Section~\ref{sec:security_analysis}. They cover substituted keys, stripped or tampered signatures and records, tampered or dropped ciphertexts and lists, and classical routes or HTLCs under the two flags of Section~\ref{sec:onion}. The fork refuses each attack without altering state or moving funds. A variant of the fork\footnote{\url{https://github.com/ahmet-kurt/pq-rust-lightning/tree/configurable}} makes the ML-DSA and ML-KEM sets selectable at build time and can replace ML-DSA with FN-DSA at degree 512 or 1024 through the fn-dsa (v0.4.0) crate. This variant serves the parameter set comparisons of Sections~\ref{sec:comp_overhead} and~\ref{sec:comm_overhead} and passes the same tests at all 15 pairings.

The network experiments of Sections~\ref{sec:latency}, \ref{sec:comm_overhead} and~\ref{sec:interop} use real Lightning nodes against Bitcoin Core v31.1\footnote{\url{https://bitcoincore.org/en/releases/31.1/}} in regtest mode. The nodes run pq-ldk-sample, our publicly available\footnote{\url{https://github.com/ahmet-kurt/pq-ldk-sample}} adaptation of LDK's reference node ldk-sample~\cite{ldksample} with commands and flags for the PQ functionality. From this source we built a \textit{vanilla} binary that links the unmodified rust-lightning code at our fork's base commit and a \textit{PQ} binary that links our fork with the feature enabled, so any difference between them comes only from PQLN. A Python harness orchestrates the Bitcoin daemon and the nodes on a workstation with a 16-core AMD Ryzen Threadripper PRO 3955WX processor and 64 GB of RAM under Ubuntu 26.04.

To assess PQLN, we use the following metrics: 1) \textit{Computational overhead} which refers to the execution time of the added cryptographic operations; 2) \textit{Payment latency} which refers to the end-to-end delay of a payment measured at the payer; 3) \textit{Communication overhead} which refers to the size increase caused by the added cryptographic material at the message and network scales; 4) \textit{Interoperability} which refers to whether PQ and vanilla nodes can coexist and transact in every combination. The vanilla build serves as the baseline for all comparisons.

\subsection{Computational Overhead Analysis}
\label{sec:comp_overhead}

We first measured the execution times of the cryptographic operations added by PQLN. A timing test in the fork calls the production functions and runs each operation 1000 times on fresh inputs in a release build after a warmup. We repeated it three times with nothing else running and averaged the results. Table~\ref{tab:timings} presents them next to vanilla rust-lightning's classical operations, and the variant supplied the FN-DSA-512 column and the other parameter-set timings below.

\begin{table}[!t]
\caption{Execution Times of the Cryptographic Operations in Microseconds (Mean and Standard Deviation of 1000 Runs)}
\label{tab:timings}
\centering
\scriptsize
\setlength{\tabcolsep}{2pt}
\begin{tabular*}{\columnwidth}{@{\extracolsep{\fill}} l r r r r @{}}
\hline
\multicolumn{5}{@{}l}{\textit{Signatures}} \\
\textbf{Operation} & \textbf{ML-DSA-44} & \textbf{FN-DSA-512} & \textbf{ECDSA} & \textbf{Schnorr} \\
Key generation & 131.4 $\pm$ 5.6 & 2006.8 $\pm$ 576.4 & \multicolumn{2}{c@{}}{19.3 $\pm$ 1.2} \\
Signing & 327.0 $\pm$ 222.5 & 215.6 $\pm$ 5.8 & 25.7 $\pm$ 1.2 & 20.3 $\pm$ 1.1 \\
Verification & 103.4 $\pm$ 6.8 & 13.5 $\pm$ 0.9 & 33.8 $\pm$ 7.1 & 33.5 $\pm$ 1.6 \\
\hline
\multicolumn{5}{@{}l}{\textit{Key exchange}} \\
\textbf{Operation} & \textbf{ML-KEM-768} & \multicolumn{3}{c@{}}{\textbf{ECDH}} \\
Key generation & 51.0 $\pm$ 1.9 & \multicolumn{3}{c@{}}{19.3 $\pm$ 1.2} \\
Encapsulation & 55.3 $\pm$ 1.9 & \multicolumn{3}{c@{}}{34.6 $\pm$ 1.6} \\
Decapsulation & 74.0 $\pm$ 2.6 & \multicolumn{3}{c@{}}{34.6 $\pm$ 1.6} \\
Hybrid secret fold & 0.4 $\pm$ 0.1 & \multicolumn{3}{c@{}}{n/a} \\
\hline
\end{tabular*}
\\[3pt]
\parbox{\columnwidth}{\footnotesize Gossip and BOLT 11 sign with ECDSA and BOLT 12 with Schnorr. All three classical schemes share one secp256k1 key generation, and one ECDH operation serves as both encapsulation and decapsulation.}
\end{table}

As can be seen from Table~\ref{tab:timings}, every operation of the default sets completes well below a millisecond. The most expensive operation is ML-DSA signing with 0.33~ms on average, which is about 13 times slower than ECDSA signing. Its large standard deviation comes from rejection sampling. From these numbers, the hybrid handshake adds around 180~$\mu$s to the initiator and 130~$\mu$s to the responder, and a hybrid onion adds 55~$\mu$s per hop to the sender and 74~$\mu$s to each hop. A BOLT 12 payee spends around 460~$\mu$s per invoice to re-derive its per-offer key and sign. The larger parameter sets stay in the same range, since ML-DSA-87 signing averages 0.64~ms. FN-DSA-512 reverses this profile. It signs and verifies faster, but its key generation takes 2.0~ms with a large variance and 9.8~ms at degree 1024. A node generates its identity key once, so only a BOLT 12 payee pays this cost repeatedly, for every invoice. These costs are small compared to the delay of a payment, which we measure next.

\subsection{Payment Latency Analysis}
\label{sec:latency}

The ciphertext list of Section~\ref{sec:onion} always carries 20 ML-KEM-768 ciphertexts and therefore adds 21.8~kB to every \textit{update\_add\_htlc}, so we measured its effect on the end-to-end delay of a payment. We built chains of one to three hops with vanilla nodes and again with PQ nodes, and the payer at one end recorded the delay from the start of route finding until it received the preimage. Every chain ran on the loopback interface and over a wide-area link emulated with netem, the Linux network emulator. The link has a 50~ms round-trip time and runs at 10~Mbit/s or at 1~Mbit/s, so the nodes pay the propagation delay and the size-dependent transmission time of an Internet path. The payer sent 50 invoice payments of 10,000 sat per configuration, and the payments on the chains of PQ nodes used pinned keys, signed invoices and hybrid onions. Table~\ref{tab:latency} presents the results.

\begin{table}[!t]
\caption{End-to-End Payment Latency in Milliseconds (Mean and Standard Deviation of 50 Payments)}
\label{tab:latency}
\centering
\scriptsize
\setlength{\tabcolsep}{3pt}
\begin{tabular}{c|cc|cc|cc}
\hline
& \multicolumn{2}{c|}{\textbf{Loopback}} & \multicolumn{2}{c|}{\textbf{50 ms RTT, 10 Mbit/s}} & \multicolumn{2}{c}{\textbf{50 ms RTT, 1 Mbit/s}} \\
\textbf{Hops} & \textbf{Vanilla} & \textbf{PQ} & \textbf{Vanilla} & \textbf{PQ} & \textbf{Vanilla} & \textbf{PQ} \\
\hline
1 & \multicolumn{1}{c|}{113 $\pm$ 44} & 98 $\pm$ 38 & \multicolumn{1}{c|}{228 $\pm$ 39} & 252 $\pm$ 49 & \multicolumn{1}{c|}{278 $\pm$ 49} & 438 $\pm$ 40 \\
2 & \multicolumn{1}{c|}{207 $\pm$ 57} & 193 $\pm$ 46 & \multicolumn{1}{c|}{418 $\pm$ 67} & 456 $\pm$ 56 & \multicolumn{1}{c|}{465 $\pm$ 65} & 818 $\pm$ 64 \\
3 & \multicolumn{1}{c|}{297 $\pm$ 54} & 306 $\pm$ 64 & \multicolumn{1}{c|}{608 $\pm$ 91} & 768 $\pm$ 84 & \multicolumn{1}{c|}{651 $\pm$ 71} & 1213 $\pm$ 79 \\
\hline
\end{tabular}
\end{table}

On loopback, a vanilla payment takes around 0.1~s per hop, and PQLN changes this by at most 15~ms. The baseline itself comes from rust-lightning, which holds every incoming HTLC and processes them in batches at random intervals averaging 56~ms. At 10~Mbit/s, PQLN adds 19 to 53~ms per hop. Transmitting the list takes 17~ms at that rate. TCP also sends at most 10 segments after an idle period, and the message spans 16 of them at the link's 1500-byte MTU, so the rest waits one round trip for an acknowledgment. At 1~Mbit/s, transmitting the list alone takes 174~ms, so PQLN adds 160 to 187~ms per hop and a three-hop payment takes 1.21~s instead of 0.65~s.

\subsection{Communication Overhead Analysis}
\label{sec:comm_overhead}

The main cost of PQLN is therefore communication. We measured it with the default parameter sets and repeated every measurement at the other NIST sets and with FN-DSA through the variant. Fig.~\ref{fig:paramsets} presents the results with filled markers for the default sets.

\begin{table*}[!t]
\caption{Interoperability Test Matrix Between Vanilla (V) and PQ (Q) Nodes (All Scenarios Pass)}
\label{tab:interop}
\centering
\scriptsize
\begin{tabular}{c c l}
\hline
\textbf{Scenario} & \textbf{Nodes} & \textbf{Functionality exercised} \\
\hline
S1 & V-V & Vanilla control with public channel, BOLT 11 both directions, keysend, BOLT 12 offer and refund, cooperative close \\
S2 & Q-Q & Same operations with all five protected surfaces, from hybrid handshake and pinning to invoice, onion and offer protections \\
S3 & Q-Q & BOLT 12 offer and refund with PQ blinded paths, covering message paths, payment paths and refund ciphertexts \\
S4 & V-Q & Mixed pair in both directions, where the vanilla node stores PQ gossip and every payment type completes classically \\
S5 & Q-Q-Q & Pin relay across the graph, two-hop hybrid onion with ciphertext rotation at the middle hop, hybrid error onion on a failure \\
S6 & Q-V-Q & Vanilla hop stores but does not relay PQ announcements, so the endpoints complete a classical payment through it \\
S7 & Q-V, Q-Q & Sender running \textit{require\_post\_quantum\_payments} refuses the classical route and still pays its pinned peer \\
S8 & V-Q, Q-Q & Receiver running \textit{require\_post\_quantum\_inbound} fails back the classical HTLC and accepts the hybrid one \\
S9 & V-V-V & Asynchronous payment control with a static invoice server and a held HTLC for an often-offline recipient \\
S10 & Q-Q-Q & ML-DSA-signed static invoice verified against the offer's committed key before the held HTLC is released \\
S11 & V-Q-Q & Vanilla sender on a PQ async offer fails closed before any HTLC is dispatched \\
S12 & Q-V-V & PQ sender pays a vanilla async offer classically \\
\hline
\end{tabular}
\end{table*}

\textit{Message Overhead:} Fig.~\ref{fig:paramsets}(a) shows how many bytes the PQ records add to each wire message, including the TLV record headers. With the default sets, a \textit{node\_announcement} grows by 4928 bytes, a \textit{channel\_update} by 2424 bytes and the two hybrid handshake acts by 2272 and 1088 bytes. The largest overhead comes from the ciphertext list of the payment onion, which adds 21,770 bytes because it always carries 20 entries. The blinded tail list adds another 10,890 bytes when a payment ends in a PQ blinded path. The gossip records grow with the ML-DSA set, and the ML-DSA-87 records exceed the 8192-byte relay budget of Section~\ref{sec:gossip} with ML-KEM-768 and ML-KEM-1024, so the variant doubles the budget for these two pairings. FN-DSA moves in the other direction, since its records add only 2759 bytes to a \textit{node\_announcement} and 670 bytes to a \textit{channel\_update} at degree 512. The ML-KEM set mainly affects the ciphertext lists, and the onion list grows from 15.4~kB with ML-KEM-512 to 31.4~kB with ML-KEM-1024, still below Lightning's 65,535-byte message limit.

\textit{Payment Requests:} Fig.~\ref{fig:paramsets}(b) shows the lengths of the payment requests generated by our nodes, and they vary by a few characters depending on the payment amount. A vanilla BOLT 11 invoice is around 400 characters, whereas it becomes 4286 characters with only the ML-DSA signature and 6396 characters when the public key is also embedded. Similarly, an offer grows from 416 to 2528 characters because of its committed key, and to 4332 characters with PQ blinded paths. A refund keeps its 531 characters because its signature stays classical, and only PQ blinded paths enlarge it to 2323 characters. The payment requests face two hard limits: 1) rust-lightning refuses to parse a BOLT 11 invoice longer than 7089 characters~\cite{ldk}, and 2) a QR code holds at most 4296 alphanumeric characters~\cite{iso18004}. Among the ML-DSA sets, only ML-DSA-44 keeps every payment request below the parser limit, and among the ML-DSA invoices only its signature-only invoice fits a QR code, with 10 characters to spare. Both FN-DSA sets stay below the parser limit, and FN-DSA-512 with ML-KEM-768 fits every request into a QR code with room to spare. ML-DSA-44 is thus the only set of the published standard whose payment requests all remain usable, which supports the choice of Section~\ref{sec:pqln_design}.

\textit{Gossip at Network Scale:} Every joining node downloads and stores the announcements of all public channels and nodes, so we measured its download and its stored graph in real networks of 5 to 100 nodes with 10 to 200 public channels. Each network has a ring topology with additional channels between random node pairs, and we built it with vanilla nodes and again with PQ nodes. After every node had learned the complete graph, a fresh observer node joined over the classical transport, so the two runs differed only in their gossip. We summed the observer's incoming gossip until its graph was complete and read the size of its graph file after it shut down.

Both costs grow linearly with the number of channels, and three repetitions produced identical byte counts because the message sizes are deterministic. Each channel adds 8111 bytes to the PQ download and 9687 bytes to the PQ graph file against 799 and 1102 bytes for vanilla, so a PQLN node downloads 10.2 times and stores 8.8 times the gossip data of a vanilla node. The two \textit{channel\_updates} of every channel dominate this cost, since their signature records account for 4848 of the 8111 bytes. Fig.~\ref{fig:paramsets}(c) shows the same per-channel costs for all 15 pairings. They rise to 14.5~kB with ML-DSA-87 and ML-KEM-1024, or 18.1 times the vanilla download, and fall to 3.3~kB with FN-DSA-512 and ML-KEM-512, or 4.2 times. A PQLN node that joins today's network of 33,000 public channels would therefore download around 270~MB and store a 320~MB graph file with the default sets, against 26~MB and 36~MB for a vanilla node. FN-DSA-512 in place of ML-DSA-44 lowers these figures to around 115~MB and 160~MB. Osuntokun projects a 29-fold growth for converting every gossip signature to ML-DSA-44~\cite{roasbeef2026pq}, whereas PQLN keeps the growth at tenfold because the largest gossip message, the \textit{channel\_announcement}, stays classical at 432 bytes.

\subsection{Interoperability Analysis}
\label{sec:interop}

Table~\ref{tab:interop} presents our 12 interoperability scenarios, which verify that PQLN nodes join the existing network and transact with vanilla nodes in every combination. Each scenario runs real nodes in the listed combination and performs the listed operations, and the node logs confirm whether every payment ran protected or classical.

S1 and S9 are the vanilla controls, and S9 to S12 cover asynchronous payments, where a static invoice server answers for an often-offline recipient and the sender holds its HTLC until the recipient returns. S2, S3, S5 and S10 exercise all five protected surfaces between PQ nodes over one and two hops.

In the mixed scenarios, a vanilla node accepts and stores the PQ announcements but does not relay them, so no pin forms across it and payments that involve it complete classically (S4, S6, S12). With the two flags of Section~\ref{sec:onion}, the sender instead refuses the classical route but still pays its pinned peer (S7), and the receiver fails back the classical HTLC but accepts the hybrid one (S8). S11 fails closed without a flag, because a vanilla sender copies the offer's metadata record into its invoice request as BOLT 12 prescribes and the enlarged request no longer fits into the fixed-size onion of an asynchronous payment, which then aborts before any HTLC is dispatched.

Every combination thus resolves to protected operation between PQ endpoints, classical operation whenever a vanilla node is involved, or a refusal before any funds move, and no payment got stuck. PQLN can therefore be deployed in today's network without a coordinated upgrade.

\section{Discussion and Limitations}
\label{sec:discussion_and_limitations}

The evaluation shows that computation is not the obstacle to PQ protection in Lightning. The cost lies in gossip instead, which grows about tenfold with ML-DSA-44 and about fourfold with FN-DSA-512. Our measurement covers the initial synchronization of a static graph, and the steady-state gossip grows as well, since every \textit{channel\_update} carries a 2424-byte signature record. Only nodes that synchronize the full graph pay this cost. For example, rust-lightning's mobile wallets instead fetch a compressed and signature-free snapshot of the graph from a server through rapid gossip sync. Such a server can add the PQ keys to the snapshot, and its clients would trust it for the keys just as they already trust it for the graph.

PQLN has several limitations. First, its key distribution relies on trust-on-first-use, so a node that first meets a peer after a cryptographically relevant quantum computer exists can pin an adversary's key. PQLN also does not yet rotate pinned keys, although a node needs this after a key compromise or for a move to FN-DSA. Two additions can close these gaps without a consensus change: 1) a \textit{node\_announcement} record signed under the old pinned key can rotate a pin, and 2) a channel's funding key can commit to its owner's PQ keys just as Taproot commits to a script tree, so a late joiner can verify these keys against a funding transaction confirmed before such a computer exists. The on-chain commitment covers only nodes with public channels, and we leave both additions to future work. Second, some signatures stay classical because their verifier holds no trusted key for the signer, namely those on the invoices of an unannounced BOLT 11 payee and on the BOLT 12 \textit{invoice\_request} and its refund response. Third, phantom payments stay classical, since they are a rust-lightning-only feature rather than part of the BOLT specifications. Finally, the TLV types, invoice tags, and feature bits of PQLN are experimental and still need assignment through the BOLT process. Our evaluation covers only rust-lightning nodes, and we defer the tests with other Lightning implementations to future work.

\section{Conclusion}
\label{sec:conclusion}

In this paper, we proposed PQLN, which brings PQ security to the off-chain surfaces of the Lightning Network without changing Bitcoin. PQLN adds ML-DSA and ML-KEM alongside the existing secp256k1 cryptography of gossip, transport, payment requests, and payment onions. It distributes the PQ keys through LN's own gossip with trust-on-first-use pinning and fits the large PQ material into the existing message formats. We analyzed the security of PQLN against a quantum adversary, implemented it in rust-lightning, and evaluated it with real Lightning nodes. The added cryptographic operations take at most 0.33~ms, and the main cost is a tenfold growth of gossip data with ML-DSA-44, or about fourfold with FN-DSA-512. Upgraded and unmodified nodes interoperate in all 12 test scenarios, so PQ protection can begin at Lightning's off-chain layer today, before Bitcoin completes its transition.

\bibliographystyle{IEEEtran}
\bibliography{references}

\end{document}